\documentclass[11pt,a4paper]{article}

\usepackage{amssymb,amsthm,amsmath,mathtools}
\usepackage{thmtools,thm-restate}
\usepackage[round]{natbib}
\usepackage{pifont}
\usepackage[usenames,svgnames,xcdraw,table]{xcolor}
\definecolor{DarkGreen}{rgb}{0.1,0.5,0.1}
\usepackage[backref=page]{hyperref}
\hypersetup{
	colorlinks=true,
	linkcolor=red,
	urlcolor=DarkGreen,
	citecolor=blue
}
\renewcommand*{\backref}[1]{}
\renewcommand*{\backrefalt}[4]{%
    \ifcase #1 (Not cited.)%
    \or        (Cited on page~#2)%
    \else      (Cited on pages~#2)%
    \fi}
\usepackage[capitalise,noabbrev]{cleveref}

\usepackage{mathrsfs,amsfonts,dsfont,authblk}
\usepackage{array,multirow,graphicx,bigdelim}

\usepackage{soul}
\usepackage{url}
\usepackage{caption}
\usepackage{graphicx}
\usepackage{amsmath}
\usepackage{amsthm}
\usepackage{amsfonts}
\usepackage{amssymb}
\usepackage{algorithm}
\usepackage[noend]{algpseudocode}
\usepackage{algorithmicx}
\usepackage{booktabs}
\usepackage{color}
\usepackage[inline]{enumitem}
\usepackage{etoolbox}
\usepackage{float}
\usepackage{mathtools}
\usepackage{multirow}
\usepackage{tabularx}
\usepackage{thmtools}
\usepackage{thm-restate}

\algrenewcommand\algorithmicindent{0.75em}

\makeatletter
\newcommand{\algmargin}{\the\ALG@thistlm}
\makeatother
\newlength{\whilewidth}
\algdef{SE}[parWHILE]{parWhile}{EndparWhile}[1]
{\parbox[t]{\dimexpr\linewidth-\algmargin}{%
\hangindent\whilewidth\strut\algorithmicwhile\ #1\ \algorithmicdo\strut}}{\algorithmicend\ \algorithmicwhile}%
\algnewcommand{\parState}[1]{\State%
\parbox[t]{\dimexpr\linewidth-\algmargin}{\strut #1\strut}}

\newtheorem{myeg}{Example}
\AtEndEnvironment{myeg}{\null\hfill\qedsymbol}
\let\oldeg\myeg
\renewcommand{\myeg}{\oldeg\normalfont}

\newtheorem{example}{Example}

\newtheorem{myq}{Q}
\let\oldq\myq
\renewcommand{\myq}{\oldq
}

\newtheorem{myprop}{Proposition}
\let\oldprop\myprop
\renewcommand{\myprop}{\oldprop
}

\newtheorem{myi}{I}
\let\oldi\myi
\renewcommand{\myi}{\oldi
}

\newcommand{\Car}[1]{\md_{#1}}

\newcommand{\crt}{\text{CRT}}

\newcommand{\ma}{{\mathcal A}}

\newcommand{\md}{{\mathcal D}}

\newcommand{\ml}{{\mathcal L}}

\newcommand{\mf}{{\mathcal F}}
\newcommand{\mn}{{\mathcal N}}
\newcommand{\mm}{{\mathcal M}}
\newcommand{\mo}{{\mathcal O}}
\newcommand{\mP}{{\mathcal P}}

\newcommand{\mtra}{\text{MTRA}}
\newcommand{\nxt}{{\text NEXT}}

\newcommand{\others}{{\text others}}
\newcommand{\pa}{\text{Pa}}

\newcommand{\newprojs}[3]{\ensuremath{{#1}_{#2\perp#3}}}
\newcommand{\projs}[2]{\ensuremath{{#1}_{\perp#2}}}
\newcommand{\proj}[3]{\ensuremath{{#1}_{\perp#2,#3}}}
\newcommand{\newproj}[4]{\ensuremath{{#1}_{#2\perp#3,#4}}}

\newcommand{\vu}{{\bf{u}}}
\newcommand{\vx}{{\bf{x}}}
\newcommand{\vy}{{\bf{y}}}
\newcommand{\vz}{{\bf{z}}}

\title{Sequential Mechanisms for Multi-type Resource Allocation}

\author[1]{Sujoy Sikdar}
\author[2]{Xiaoxi Guo}
\author[2]{Haibin Wang}
\author[3]{Lirong Xia}
\author[4]{Yongzhi Cao}

\affil[1]{Binghamton University\\
	{\small\texttt{ssikdar@binghamton.edu}}}
\affil[2]{Peking University\\
	{\small\texttt{guoxiaoxi@pku.edu.cn, beach@pku.edu.cn}}}
\affil[3]{Rensselaer Polytechnic Institute\\
	{\small\texttt{xial@cs.rpi.edu}}}
\affil[4]{Peking University\\
	{\small\texttt{caoyz@pku.edu.cn}}}

\begin{document}

\maketitle

\begin{abstract}
Several resource allocation problems involve multiple types of resources, with a different agency being responsible for ``locally'' allocating the resources of each type, while a central planner wishes to provide a guarantee on the properties of the final allocation given agents' preferences. We study the relationship between properties of the local mechanisms, each responsible for assigning all of the resources of a designated type, and the properties of a {\em sequential mechanism} which is composed of these local mechanisms, one for each type, applied sequentially, under {\em lexicographic preferences}, a well studied model of preferences over multiple types of resources in artificial intelligence and economics. We show that when preferences are $O$-legal, meaning that agents share a common importance order on the types, sequential mechanisms satisfy the desirable properties of anonymity, neutrality, non-bossiness, or Pareto-optimality if and only if every local mechanism also satisfies the same property, and they are applied sequentially according to the order $O$. Our main results are that under $O$-legal lexicographic preferences, every mechanism satisfying strategyproofness and a combination of these properties must be a sequential composition of local mechanisms that are also strategyproof, and satisfy the same combinations of properties.
\end{abstract}

\section{Introduction}
Consider the example of a hospital where patients must be allocated surgeons and nurses with different specialties, medical equipment of different types, and a room~\cite{Huh13:Multiresource}. This example illustrates multi-type resource allocation problems (\mtra{}s), first introduced by~\citet{Moulin95:Cooperative}, where there are $p\ge 1$ types of {\em indivisible} items which are not interchangeable, and a group of agents having heterogeneous preferences over receiving combinations of an item of each type. The goal is to design a mechanism which allocates each agent with a {\em bundle} consisting of an item of each type. 

Often, a different agency is responsible for the allocation of each type of item in a distributed manner, using possibly different {\em local} mechanisms, while a central planner wishes that the mechanism composed of these local mechanisms satisfies certain desirable properties. For example, different departments may be responsible for the allocation of each type of medical resources, while the hospital wishes to deliver a high standard of patient care and satisfaction given the patients' preferences and medical conditions; in an enterprise, clients have heterogeneous preferences over cloud computing resources like computation and storage~\cite{Ghodsi11:Dominant,Ghodsi12:Multi,Bhattacharya13:Hierarchical}, possibly provided by different vendors; in a university, students must be assigned to different types of courses handled by different departments; in a seminar class, the research papers and time slots~\cite{Mackin2016:Allocating} may be assigned separately by the instructor and a teaching assistant respectively, and in rationing~\cite{elster1992local}, different agencies may be responsible for allocating different types of rations such as food and shelter.

Unfortunately, as~\citet{Svensson99:Strategy-proof} shows, even when there is a single type of items and each agent is to be assigned a single item, serial dictatorships are the only strategyproof mechanisms which are {\em non-bossy}, meaning that no agent can falsely report her preferences to change the outcome without also affecting her own allocation, and {\em neutral}, meaning that the outcome is independent of the names of the items. In a serial dictatorship, agents are assigned their favorite remaining items one after another according to a fixed priority ordering of the agents.~\citet{Papai01:Strategyproof} shows a similar result for the {\em multiple assignment problem}, where agents may be assigned more than one item, that the only mechanisms which are strategyproof, non-bossy, and {\em Pareto-optimal} are {\em sequential dictatorships}, where agents pick a favorite remaining item one at a time according to a hierarchical picking sequence, where the next agent to pick an item depends only on the allocations made in previous steps. Pareto-optimality is the property that there is no other allocation which benefits an agent without making at least one agent worse off. More recently,~\citet{hosseini2019multiple} show that even under lexicographic preferences, the only mechanisms for the multiple assignment problem that are strategyproof, non-bossy, neutral and Pareto-optimal are serial dictatorships with a quota for each agent.

~\citet{Mackin2016:Allocating} study \mtra{}s in a slightly different setting to ours: a monolithic central planner controls the allocation of all types of items. They characterize serial dictatorships under the unrestricted domain of strict preferences over bundles with strategyproofness, non-bossiness, and type-wise neutrality, a weaker notion of neutrality where the outcome is independent of permutations on the names of items within each type. Perhaps in light of this and other negative results described above, there has been little further work on strategyproof mechanisms for \mtra{}s. This is the problem we address in this paper.

We study the design of strategyproof sequential mechanisms for \mtra{}s with $p\ge 2$ types, which are composed of $p$ local mechanisms, one for each type, applied sequentially one after the other, to allocate all of the items of the type, under the assumption that agents' preferences are {\em lexicographic} and $O$-legal.

For ~\mtra{}s, lexicographic preferences are a natural, and well-studied assumption for reasoning about ordering alternatives based on multiple criteria in social science~\cite{Gigerenzer96:Reasoning}. In artificial intelligence, lexicographic preferences have been studied extensively, for \mtra{}s~\cite{Sikdar2017:Mechanism,sikdar2019mechanism,Sun2015:Exchange,wang2020multi,guo2020probabilistic}, multiple assignment problems~\cite{hosseini2019multiple,Fujita2015:A-Complexity}, voting over multiple issues~\cite{Lang09:Sequential,Xia11:Strategic}, and committee selection~\cite{sekar2017condorcet}, since lexicographic preferences allow reasoning about and representing preferences in a structured and compact manner. In \mtra{}s, lexicographic preferences are defined by an importance order over the types of items, and local preferences over items of each type. The preference relation over any pair of bundles is decided in favor of the bundle that has the more preferred item of the most important type at which the pair of bundles differ, and this decision depends only on the items of more important types.

In several problems, it is natural to assume that every agent shares a common importance order. For example, when rationing~\cite{elster1992local}, it may be natural to assume that every agent thinks food is more important than shelter, and in a hospital~\cite{Huh13:Multiresource}, all patients may consider their allocation of surgeons and nurses to be more important than the medical equipment and room. $O$-legal lexicographic preference profiles, where every agent has a common importance order $O$ over the types, have been studied recently by~\cite{Lang09:Sequential,Xia11:Strategic} for the multi-issue voting problem. When agents' preferences are $O$-legal and lexicographic, it is natural to ask the following questions about sequential mechanisms that decide the allocation of each type sequentially using a possibly different local mechanism according to $O$, which we address in this paper: (1) {\em if every local mechanism satisfies property $X$, does the sequential mechanism composed of these local mechanisms also satisfy $X$?}, and (2) {\em what properties must every local mechanism satisfy so that their sequential composition satisfies property $X$?}

\subsection{Contributions}
For $O$-legal preferences, a property $X\in\{$anonymity, type-wise neutrality, non-bossiness, monotonicity, Pareto-optimality$\}$, and any sequential mechanism $f_O=(f_1,\dots,f_p)$ which applies each local mechanism $f_i$ one at a time according to the importance order $O$, we show in~\Cref{thm:localtoglobal} and~\Cref{thm:globaltolocal} that $f_O$ satisfies $X$ if and only if every local mechanism it is composed of satisfies $X$.

However, sequential compositions of locally strategyproof mechanisms are not guaranteed to be strategyproof, which raises the question: under what conditions are sequential mechanisms strategyproof? We begin by showing in~\Cref{prop:localspnotOlegal}, that when agents preferences are lexicographic, but agents have different importance orders, sequential mechanisms composed of locally strategyproof mechanisms are, unfortunately, not guaranteed to be strategyproof. In contrast, we show in~\Cref{prop:localsp2sp} that sequential composition of strategyproof mechanisms are indeed strategyproof when either: (1) agents' preferences are separable and lexicographic, even when different agents may have different importance orders, or (2) agents' preferences are lexicographic and $O$-legal and all of the local mechanisms are also non-bossy.

Our main results characterize the class of mechanisms that satisfy strategyproofness, along with different combinations of non-bossiness, neutrality, and Pareto-optimality under $O$-legal preferences as $O$-legal sequential mechanisms. We show:
\begin{itemize}[wide,topsep=0pt,labelindent=0pt]
\item In~\Cref{thm:crnets}, that under $O$-legal lexicographic preferences, the class of mechanisms satisfying strategyproofness and non-bossiness is exactly the class of mechanisms that can be {\em decomposed} into multiple locally strategyproof and non-bossy mechanisms, one for each combination of type and allocated items of more important types. This class of mechanisms is exactly the class of $O$-legal {\em conditional rule nets} (CR-nets)~\cite{Lang09:Sequential};
\item In~\Cref{thm:char}, that a mechanism is strategyproof, non-bossy, and type-wise neutral if and only if it is an $O$-legal sequential composition of serial dictatorships;
\item In~\Cref{thm:char2}, that a mechanism is strategyproof, non-bossy, and Pareto-optimal if and only if it is an $O$-legal CR-net composed of serial dictatorships. 
\end{itemize}

Finally, we show that despite the negative result in~\Cref{prop:localspnotOlegal} that when agents' preferences do not share a common importance order on the types, sequential compositions of locally strategyproof mechanisms may not satisfy strategyproofness, we show in Theorem~\ref{thm:manipulationhard}, that computing beneficial manipulations w.r.t. a sequential mechanism is NP-complete.

\section{Related Work and Discussion}
The \mtra{} problem was introduced by~\citet{Moulin95:Cooperative}. More recently, it was studied by~\citet{Mackin2016:Allocating}, who characterize the class of strategyproof and non-bossy mechanisms under the unrestricted domain of strict preferences over bundles as the class of serial dictatorships. However, as they note, it may be unreasonable to expect agents to express preferences as complete rankings over all possible bundles, besides the obvious communication and complexity issues arising from agents' preferences being represented by complete rankings.

The literature on characterizations of strategyproof mechanisms~\cite{Svensson99:Strategy-proof,Papai01:Strategyproof,hosseini2019multiple} for resource allocation problems belong to the line of research initiated by the famous Gibbard-Satterthwaite Theorem~\cite{Gibbard73:Manipulation,Satterthwaite75:Strategy} which showed that dictatorships are the only strategyproof voting rules which satisfy non-imposition, which means that every alternative is selected under some preference profile. Several following works have focused on circumventing these negative results by identifying reasonable and natural restrictions on the domain of preferences. For voting,~\cite{Moulin80:Strategy} provide non-dictatorial rules satisfying strategyproofness and non-imposition under single-peaked~\cite{Black48:Rationale} preferences. Our work follows in this vein and is closely related to the works by~\citet{LeBreton99:Separable}, who assume that agents' preferences are separable, and more recently,~\citet{Lang09:Sequential} who consider the multi-issue voting problem under the restriction of $O$-legal lexicographic preferences, allowing for conditional preferences given by CP-nets similar to our work.~\citet{Xia10:Strategy} consider a weaker and more expressive domain of lexicographic preferences allowing for conditional preferences. Here, agents have a common importance order on the issues, and the agents preferences over any issue is conditioned only on the outcome of more important issues. They characterize the class of voting rules satisfying strategyproofness and non-imposition as being exactly the class of all CR-nets. CR-nets define a hierarchy of voting rules, where the voting rule for the most important issue is fixed, and the voting rule for every subsequent issue depends only on the outcome of the previous issues. Similar results were shown earlier by~\cite{Barbera93:Generalized,Barbera97:Voting,Barbera91:Voting}.

In a similar vein, ~\citet{Sikdar2017:Mechanism,sikdar2019mechanism} consider the multi-type variant of the classic housing market~\cite{Shapley74:Cores}, first proposed by~\citet{Moulin95:Cooperative}, and~\citet{Fujita2015:A-Complexity} consider the variant where agents can receive multiple items. 
These works circumvent previous negative results on the existence of strategyproof and core-selecting mechanisms under the assumption of lexicographic extensions of CP-nets, and lexicographic preferences over bundles consisting of multiple items of a single type respectively.~\citet{wang2020multi,guo2020probabilistic} study \mtra{}s with divisible and indivisible items, and provide mechanisms that are fair and efficient under the notion of stochastic dominance by extending the famous probabilistic serial~\cite{Bogomolnaia01:New} and random priority~\cite{Abdulkadiroglu98:Random} mechanisms, and show that while their mechanisms do not satisfy strategyproof in general, under the domain restriction of lexicographic preferences, strategyproofness is restored, and stronger notions of efficiency can be satisfied.



\section{Preliminaries}
A {\em multi-type resource allocation problem (\mtra{})}~\cite{Mackin2016:Allocating}, is given by a tuple $(N,M,P)$. Here, \begin{enumerate*}[label=(\arabic*)]\item $N=\{1,\dots,n\}$ is a set of agents, \item $M=D_1\cup\dots\cup D_p$ is a set of items of $p$ types, where for each $i\le p$, $D_i$ is a set of $n$ items of type $i$, and \item $P=(\succ_j)_{j\le n}$ is a {\em preference profile}, where for each $j\le n$, $\succ_j$ represents the preferences of agent $j$ over the set of all possible {\em bundles} $\md=D_1\times\dots\times D_p$\end{enumerate*}. For any type $i\le p$, we use $k_i$ to refer to the $k$-th item of type $i$, and we define $T=\{D_1,\dots,D_p\}$. 
We also use $D_{<i}$ to refer to the set of $\{D_1,\dots,D_{i-1}\}$, $D_{>i}$ refers to $\{D_{i+1},\dots,D_p\}$, and $D_{\le i},D_{\ge i}$ are in the same manner.
For any profile $P$, and agent $j\le n$, we define $P_{-j}=(\succ_k)_{k\le n, k\neq j}$, and $P=(P_{-j},\succ_j)$.

\paragraph{Bundles.} Each bundle $\vx\in\md$ is a $p$-vector, where for each type $i\le p$, $[\vx]_i$ denotes the item of type $i$. We use $a \in \vx$ to indicate that bundle $\vx$ contains item $a$. For any $S\subseteq T$, we define $\Car{S}=\times_{D\in S}D$, and $-S=T\setminus S$. For any $S\subseteq T$, any bundle $\vx\in\Car{S}$, for any $D\in -S$, and item $a\in D$, $(a,\vx)$ denotes the bundle consisting of $a$ and the items in $\vx$, and similarly, for any $U\subseteq -S$, and any bundle $\vy\in\Car{U}$, we use $(\vx,\vy)$ to represent the bundle consisting of the items in $\vx$ and $\vy$. For any $S\subseteq T$, we use $\projs{\vx}{S}$ to denote the items in $\vx$ restricted to the types in $S$.

\paragraph{Allocations.}
An {\em allocation} $A:N\to\md$ is a one-to-one mapping from agents to bundles such that no item is assigned to more than one agent. $\ma$ denotes the set of all possible allocations. Given an allocation $A\in\ma$, $A(j)$ denotes the bundle allocated to agent $j$. For any $S\subseteq T$, we use $\projs{A}{S}:N\to\Car{S}$ to denote the allocation of items restricted to the types in $S$.

\paragraph{CP-nets and $O$-legal Lexicographic Preferences.}
An {\em acyclic CP-net}~\cite{Boutilier04:CP} $\mn$ over $\md$ is defined by \begin{enumerate*}[label=(\roman*)]\item a directed graph $G=(T,E)$ called the {\em dependency graph}, and \item for each type $i\le p$, a {\em conditional preference table} $CPT(D_i)$ that contains a linear order $\succ^{\vx}$ over $D_i$ for each $\vx\in\Car{Pa(D_i)}$, where $Pa(D_i)$ is the set of types corresponding to the parents of $D_i$ in $G$.\end{enumerate*} A CP-net $\mn$ represents a partial order over $\md$ which is the transitive closure of the preference relations represented by all of the $CPT$ entries which are $\{(a_i,\vu,\vz) \succ (b_i,\vu,\vz):i \le p, a_i,b_i\in D_i,\vu\in\Car{Pa(D_i)},\vz\in\Car{-Pa(D_i)\setminus \{D_i\}}\}$.

Let $O=[D_1\rhd\dots\rhd D_p]$ be a linear order over the types. A CP-net is $O$-legal if there is no edge $(D_i,D_l)$ with $i>l$ in its dependency graph. A {\em lexicographic} extension of an $O$-legal CP-net $\mn$ is a linear order $\succ$ over $\md$, such that for any $i\le p$, any $\vx\in\md_{D_{<i}}$, any $a_i,b_i\in D_i$, and any $\vy,\vz\in \md_{D_{>i}}$, if $a_i\succ^{\vx}b_i$ in $\mn$, then, $(\vx,a_i,\vy)\succ (\vx,b_i,\vz)$. The linear order $O$ over types is called an {\em importance order}, and $D_1$ is the most important type, $D_2$ is the second most important type, etc. We use $\mo$ to denote the set of all possible importance orders over types.

Given an important order $O$, we use $\ml_{O}$ to denote the set of all possible linear orders that can be induced by lexicographic extensions of $O$-legal CP-nets as defined above. A preference relation $\succ\in\ml_{\mo}$ is said to be an $O$-legal lexicographic preference relation, and a profile $P\in\ml_{O}^n$ is an $O$-legal lexicographic profile. An $O$-legal preference relation is {\em separable}, if the dependency graph of the underlying CP-net has no edges. We will assume that all preferences are $O$-legal lexicographic preferences throughout this paper unless specified otherwise.
\begin{figure}
  \centering
  \includegraphics[width=\linewidth]{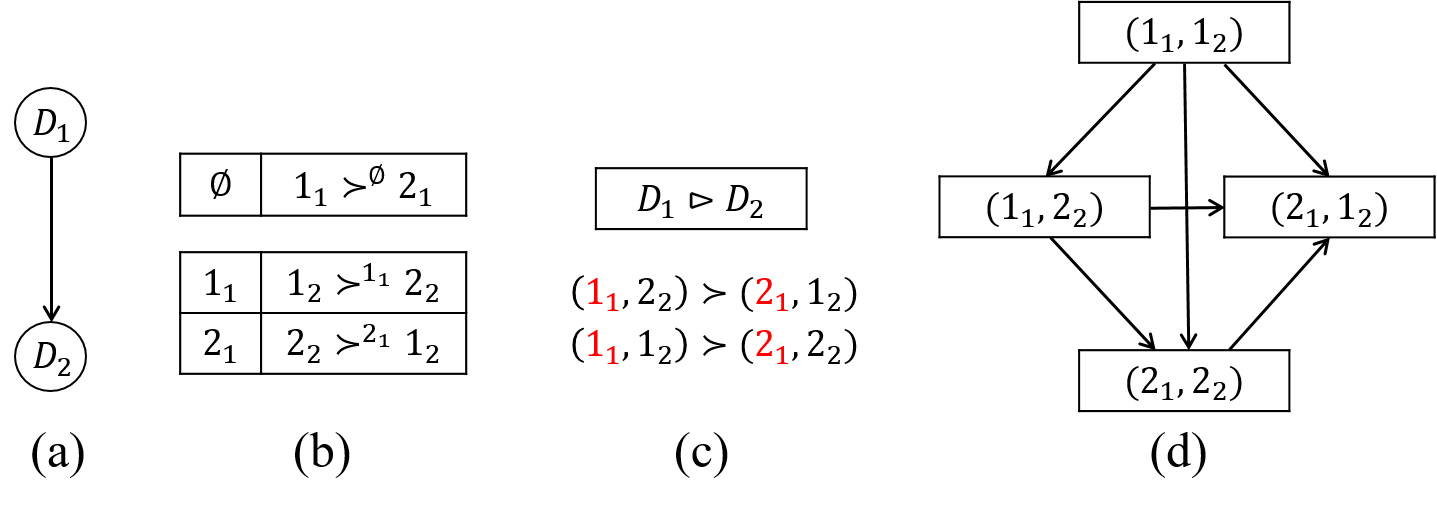}
  \caption{An $O$-legal lexicographic preference with an underlying CP-net, where $O=[D_1\rhd D_2]$.}
  \label{eg:pref}
\end{figure}
\begin{example}
Here we show how to compare bundles under an $O$-legal lexicographic preference with CP-net.
In Figure~\ref{eg:pref}(a) is a dependency graph which shows that $D_2$ depends on $D_1$.
Figure~\ref{eg:pref}(b) is the $CPT$ for both types, which implies $(1_1,1_2)\succ(1_1,2_2),(2_1,2_2)\succ(2_1,1_2)$.
Figure~\ref{eg:pref}(c) gives the importance order $O=[D_1\rhd D_2]$.
With $O$ we can compare some bundles directly.
For example, $(1_1,2_2)\succ(2_1,1_2),(1_1,1_2)\succ(2_1,2_2)$ because the most important type with different allocations is $D_1$ and $1_1 \succ^{\emptyset}2_1$.
Finally, Figure~\ref{eg:pref}(d) shows the relations among all the bundles.
\end{example}
We note that any lexicographic extension of an $O$-legal CP-net according to the order $O$ does not violate any of the relations induced by the original CP-net, and always induces a linear order over all possible bundles unlike CP-nets which may induce partial orders.

For any $O$-legal lexicographic preference relation $\succ$ over $\md$, and given any $\vx\in\md_{D_{<i}}$, we use $\proj{\succ}{D_i}{\vx}$ as the {\em projection} of the relation $\succ$ over $D_i$ given $\vx$, and $\proj{\succ}{D_{\ge i}}{\vx}$ as the projection of $\succ$ over $\{(\vx,\vz):\vz\in\md_{D_{\ge i}}\}$. For convenience, given an allocation $A$, for any $i\le p$, we define $\proj{\succ}{D_i}{A}$ and $\proj{\succ}{D_{\ge i}}{A}$ similarly, where the preferences are projected based on the allocation of items of types that are more important than $i$, and given an $O$-legal lexicographic profile $P$, we define $\proj{
P}{D_i}{A}$ and $\proj{P}{D_{\ge i}}{A}$ similarly, by projecting the preferences of every agent. We just leave out $\vx$ (and similarly, $A$) if $i=1$. We use $D_{-i}$ to stand for the set of all types except $D_i$.


\paragraph{Sequential and Local Mechanisms.}
An \emph{allocation mechanism} $f:\mP\to\ma$ maps $O$-legal preference profiles to allocations. Given an importance order $O=[D_1\rhd\dots\rhd D_p]$, an $O$-legal sequential mechanism $f_O=(f_1,\dots,f_p)$ is composed of $p$ {\em local} mechanisms, that are applied one after the other in $p$ rounds, where in each round $i\le p$, a local mechanism $f_i$ allocates all of the items of $D_i$ given agents' projected preferences over $D_i$ conditioned on the partial allocation in previous rounds. 

\paragraph{Desirable Properties.} An allocation mechanism $f$ satisfies:
\begin{itemize}[wide,topsep=0pt,labelindent=0pt]
\item \emph{anonymity}, if for any permutation $\Pi$ on the names of agents, and any profile $P$, $f(\Pi(P)) = \Pi(f(P))$;
\item \emph{type-wise neutrality}, if for any permutation $\Pi=(\Pi_1,\dots,\Pi_p)$, where for any $i\le p$, $\Pi$ only permutes the names of the items of type $i$ according to a permutation $\Pi_i$, and any profile $P$, $f(\Pi(P)) = \Pi(f(P))$;
\item \emph{Pareto-optimality}, if for every allocation $A$ such that there exists an agent $j$ such that $A(j)\succ_j f(P)(j)$, there is another agent $k$ such that $f(P)(k)\succ_j A(k)$.
\item \emph{non-bossiness}, if no agent can misreport her preferences and change the allocation of other agents without also changing her own allocation, i.e. there does not exist any pair $(P,\succ'_j)$ where $P$ is a profile and $\succ'_j$ is the misreported preferences of agent $j$ such that $f(P)(j) = f(P_{-j},\succ'_j)(j)$ and for some agent $k \neq j$, $f(P)(k) \neq f(P_{-j},\succ'_j)(k)$.
\item \emph{non-bossiness of more important types}, if no agent $j$ can misreport her local preferences for less important types and change the allocation of more important types to other agents without also changing her own allocation of more important types. i.e. for every profile $P$, every agent $j \le n$, every type $D_i, i \le p$, and every misreport of agent $j$'s preferences $\succ'_j$ where for every $h < i$, every $u \in Par(D_{h})$, $\newproj{\succ'}{j}{D_{h}}{u} = \newproj{\succ}{j}{D_{h}}{u}$, it holds that if for some agent $k \neq j$, $\projs{f(P_{-j},\succ'_j)(k)}{D_{\le i}} \neq \projs{f(P)(k)}{D_{\le i}}$, then $\projs{f(P_{-j},\succ'_j)(j)}{D_{\le i}} \neq \projs{f(P)(j)}{D_{\le i}}$.

\item \emph{monotonicity}, for any agent $j$, any profile $P$, let $\succ'_j$ be a misreport preference such that if $Y\subseteq\md$ is the set of all bundles whose ranks are raised and it holds that for every $\vx,\vz\in Y$, $\vx\succ_j\vz \implies \vx\succ'_j\vz$, then, $f(P_{-j}, \succ'_j)(j)\in \{f(P)(j)\}\cup Y$.



\item \emph{strategyproofness}, if no agent has a beneficial manipulation i.e. there is no pair $(P,\succ'_j)$ where $P$ is a profile and $\succ'_j$ is a manipulation of agent $j$'s preferences such that $f(P_{-j},\succ'_j)(j) \succ_j f(P)(j)$.
\end{itemize}

\section{Properties of Sequential Mechanisms Under Lexicographic Preferences}

\begin{restatable}{thm}{thmlocaltoglobal}\label{thm:localtoglobal}
	For any importance order $O\in\mo$, any $X \in \{$anonymity, type-wise neutrality, non-bossiness, monotonicity, Pareto-optimality$\}$, and $f_O=(f_1,\dots,f_p)$ be any $O$-legal sequential mechanism. Then, for $O$-legal preferences, if for every $i\le p$, the local mechanism $f_{i}$ satisfies $X$, then $f_O$ satisfies $X$.
\end{restatable}
\begin{proof}(Sketch)
	Throughout, we will assume that $O=[D_1\rhd\dots\rhd D_p]$, and that $P$ is an arbitrary $O$-legal preference profile over $p$ types. For any $i\le p$, we define $g_i$ to be the sequential mechanism $(f_1,\dots,f_i)$. The proofs of anonymity and type-wise neutrality are relegated to the appendix.
	
	\noindent{\bf non-bossiness.}
	Let us assume for the sake of contradiction that the claim is false, i.e. there exists a profile $P$, an agent $j$ and a misreport $\succ'_j$ such that for $P' = (\succ_{-j},\succ'_{j})$, $f_O(P)(j) = f_O(P')(j)$, and $f_O(P) \neq f_O(P')$. Then, there is a type $i\le p$ such that, $\projs{f_O(P)}{D_{<i}} = \projs{f_O(P')}{D_{<i}}$ and $\projs{f_O(P)}{D_i} \neq \projs{f_O(P')}{D_i}$. Let $A = \projs{f_O(P)}{D_{<i}}$. Then, there is an agent $k$ such that $f_{i}(\proj{P}{D_i}{A})(k)\neq f_{i}(\proj{P'}{D_i}{A})(k)$. 
	By the choice of $i$, and the assumption that every other agent reports preferences truthfully, $\proj{\succ_j}{D_i}{A}\neq \proj{\succ'_j}{D_i}{A}$. Then, $f_i(\proj{\succ_{-j}}{D_i}{A},\proj{\succ_j}{D_i}{A})(j)= f_i(\proj{\succ_{-j}}{D_i}{A},\proj{\succ'_j}{D_i}{A})(j)$, but $f_i(\proj{\succ_{-j}}{D_i}{A},\proj{\succ_j}{D_i}{A})(k)\neq f_i(\proj{\succ_{-j}}{D_i}{A},\proj{\succ'_j}{D_i}{A})(k)$, a contradiction to our assumption that $f_i$ is non-bossy.

	\noindent{\bf monotonicity.} 
	Let $P'=(P_{-j},\succ'_j)$ be an $O$-legal profile obtained from $P$ and $Y\subseteq\md$ is the set of bundles raising the ranks in $P'$ such that the relative rankings of bundles in $Y$ are unchanged in $P$ and $P'$. 
	For any $Y\subseteq\md$, and any $\vu\in\md_{D_{<i}}$, let $Y^{D_i\mid \vu} = \{x_i: \vx\in Y,x_h=u_h \text{ for all } h\le i-1\}$. 
	It is easy to see that if $\vx_1 = \projs{f_O(P')(j)}{\{D_1\}}$, then it follows from strong monotonicity of $f_1$ that $\vx_1\in\projs{f_O(P)(j)}{\{D_1\}}\cup Y^{D_1}$.
	Now, either $\vx_1\neq\projs{f_O(P)(j)}{\{D_1\}}$, or $\vx_1=\projs{f_O(P)(j)}{\{D_1\}}$. 
	Suppose $\vx_1\neq\projs{f_O(P)(j)}{\{D_1\}}$. 
	Then, by strong monotonicity of $f_1$, $\vx_1\succ\projs{f_O(P)(j)}{\{D_1\}}$. 
	Then, by our assumption of $O$-legal lexicographic preferences, for any $\vz\in\md_{\{D_2,\dots,D_p\}}$, $(\vx_1,\vz)\in Y$. 
	Therefore, $f_O(P')(j)\in Y$. Suppose $\vx_1=\projs{f_O(P)(j)}{\{D_1\}}$, then by a similar argument, $\projs{f_O(P')(j)}{\{D_2\}}\in \{\projs{f_O(P)(j)}{\{D_2\}}\}\cup Y^{D_2\mid (\vx_1)}$. 
	Applying our argument recursively, we get that $f_O(P')(j)\in \{f_O(P)(j)\}\cup Y$.
	
	\noindent{\bf Pareto-optimality.}
	Suppose the claim is true for $p\le k$ types. Let $P$ be an $O$-legal lexicographic profile over $k+1$ types, and $f_O=(f_i)_{i\le k+1}$ is a sequential composition of Pareto-optimal local mechanisms. Suppose for the sake of contradiction that there exists an allocation $B$ such that some agents strictly better off compared to $f_O(P)$, and no agent is worse off. Then, by our assumption of lexicographic preferences, for every agent $k$ who is not strictly better off, $B(k) = f_O(P)(k)$, and for every agent $j$ who is strictly better off, one of two cases must hold. (1) $\projs{B(j)}{D_1}\succ_j \projs{f_O(P)(j)}{D_1}$, or (2) $\projs{B(j)}{D_1} = \projs{f_O(P)(j)}{D_1}$. (1): If there exists an agent such that $\projs{B(j)}{D_1}\succ_j \projs{f_O(P)(j)}{D_1}$, this is a contradiction to our assumption that $f_1$ is Pareto-optimal. (2): Suppose $\projs{B(j)}{D_1}=\projs{f_O(P)(j)}{D_1}$ for all agents who are strictly better off. Let $g=(f_2,\dots,f_{k+1})$. W.l.o.g. let agent $1$ strictly prefer $B(1)$ to $f_O(P)(1)$. Then, $g(\proj{P}{D_{\le k+1}\setminus D_1}{\projs{f_O(P)}{D_1}})(1)\succ_1 \projs{B(1)}{D_{\le k+1}\setminus D_1}$, and for every other agent $l\neq 1$, either $g(\proj{P}{D_{\le k+1}\setminus D_1}{\projs{f_O(P)}{D_1}})(l) \succ_l \projs{B(l)}{D_{\le k+1}\setminus D_1}$, or $g(\proj{P}{D_{\le k+1}\setminus D_1}{\projs{f_O(P)}{D_1}})(l) =  \projs{B(l)}{D_{\le k+1}\setminus D_1}$, which is a contradiction to our induction assumption.
\end{proof}

\begin{restatable}{thm}{thmglobaltolocal}\label{thm:globaltolocal}
	For any importance order $O\in\mo$,  $X \in \{$anonymity, type-wise neutrality, non-bossiness, monotonicity, Pareto-optimality$\}$, and $f_O=(f_1,\dots,f_p)$ be any $O$-legal sequential mechanism. For $O$-legal preferences, if $f_O$ satisfies $X$, then for every $i\le p$, $f_{i}$ satisfies $X$.
\end{restatable}
\begin{proof}(Sketch) We only provide the proof of non-bossiness here. The rest of the proofs are in the appendix.
    
    \noindent{\bf non-bossiness.}
	Assume for the sake of contradiction that $k\le p$ is the most important type such that $f_{k}$ does not satisfy non-bossiness. Then, there exists a preference profile $Q = (\succ^k)_{j \le n}$ over $D_{k}$, and a bossy agent $l$ and a misreport $Q' = (\succ^k_{-l}, \bar\succ^k_{l})$, such that  $f_k(Q')(l) = f_k(Q)(l)$, but $f_{k}(Q') \neq f_{k}(Q)$. Now, consider the $O$-legal separable lexicographic profile $P$, where for any type $i\le p$, the preferences over type $D_i$ is denoted $\projs{P}{D_i}$ and $\projs{P}{D_k}=Q$, and the profile $P'$ obtained from $P$ by replacing $\succ_l$ with $\succ'_l$, which in turn is obtained from $\succ_l$ by replacing $\projs{\succ_l}{D_k}$ with $\bar\succ^k_l$. It is easy to see that $\projs{f_O(P')}{D_{<k}}=\projs{f_O(P)}{D_{<k}}$, and $\projs{f_O(P')(l)}{D_k} = \projs{f_O(P)(l)}{D_k}$, but $\projs{f_O(P')}{D_k} \neq \projs{f_O(P)}{D_k}$, and by our assumption of separable preferences, $\projs{f_O(P')}{D_{>k}}=\projs{f_O(P)}{D_{>k}}$. This implies that $f_O(P')(l)=f_O(P)(l)$, but $f_O(P')\neq f_O(P)$, implying that $f_O$ does not satisfy non-bossiness, which is a contradiction.
\end{proof}

\section{Strategyproofness of Sequential Mechanisms}

A natural question to ask is whether it is possible to design strategyproof sequential mechanisms when preferences are lexicographic, but each agent $j\le n$ may have a possibly different importance order $O_j\in\mo$ over the types, and their preference over $\md$ is $O_j$-legal and lexicographic. A sequential mechanism applies local mechanisms according to some importance order $O\in\mo$ and is only well defined for $O$-legal preferences. When preferences are not $O$-legal, it is necessary to define how to project agents' preferences given a partial allocation when a sequential mechanism is applied. Consider an agent $j$ with $O_j$-legal lexicographic preferences, and a partial allocation $\projs{A}{S}$ for some 
$S\subseteq T$, which allocates $\vx\in \md_S$ to $j$. A natural question to ask is how should agent $j$'s preferences be interpreted over a type $D_i$ which has not been allocated yet. We define two natural ways in which agents' may wish their preferences to be interpreted. We say that an agent is {\em optimistic}, if for any type $D_i\not\in S$, and any pair of items $a_i,b_i\in D_i$, $a_i\succ b_i$ if and only if according to their original preferences $\sup\{\vy\in\md:\vy_k = \vx_k \text{ for every } D_k\in S, \vy_i=a_i\} \succ \sup\{\vy\in\md:\vy_k = \vx_k \text{ for every } D_k\in S, \vy_i=b_i\}$. Similarly, an agent is {\em pessimistic}, if for any type $D_i\not\in S$, and any pair of items $a_i,b_i\in D_i$, $a_i\succ b_i$ if and only if $\inf\{\vy\in\md:\vy_k = \vx_k \text{ for every } D_k\in S, \vy_i=a_i\} \succ \inf\{\vy\in\md:\vy_k = \vx_k \text{ for every } D_k\in S, \vy_i=b_i\}$.

\begin{restatable}{prop}{proplocalspnotOlegal}\label{prop:localspnotOlegal}
For any importance order $O\in\mo$, when the preferences are not $O$-legal, and agents are either optimistic or pessimistic, a sequential mechanism $f_O$ composed of strategyproof mechanisms is not necessarily strategyproof.
\end{restatable}
\begin{proof}
    When preferences are lexicographic, and not $O$-legal, a sequential  mechanism composed of locally strategyproof mechanisms is not necessarily strategyproof, when agents are either optimistic or pessimistic, as we show with counterexamples. Consider the profile with two agents and two types $H$ and $C$. Agent $1$'s importance order is $H \rhd C$, preferences over $H$ is $1_H \succ 2_H$ and over $C$ is conditioned on the assignment of house $1_H: 1_C \succ 2_C, 2_H: 2_C \succ 1_C$. Agent $2$ has importance order $C \rhd H$ and separable preferences with order on cars being $2_C \succ 1_C$, and order on houses $1_H \succ 2_H$. Consider the sequential mechanism composed of serial dictatorships where $H \rhd C$ and for houses the picking order over agents is $(2,1)$, and for cars $(1,2)$. When agents are truthful and either optimistic or pessimistic, the allocation is $2_H2_C$ and $1_H1_C$ respectively to agents $1$ and $2$. When agent $2$ misreports her preferences over houses as $2_H \succ 1_H$, and agent $1$ is truthful and either optimistic or pessimistic, the allocation is $1_H1_C$ and $2_H2_C$ to agents $1$ and $2$ respectively, a beneficial misreport for agent $2$.
\end{proof}

In contrast, sequential mechanisms composed of locally strategyproof mechanisms are guaranteed to be strategyproof under two natural restrictions on the domain of lexicographic preferences: (1) when agents' preferences are lexicographic and separable, but not necessarily $O$-legal w.r.t. a common importance order $O$, and (2) when agents' have $O$-legal lexicographic preferences, and the local mechanisms also satisfy non-bossiness.

\begin{restatable}{prop}{proplocalsptosp}\label{prop:localsp2sp}
	For any importance order $O\in\mo$, a sequential mechanism composed of strategyproof local mechanisms is strategyproof,
	\begin{enumerate}[label=(\arabic*),wide,labelindent=0pt]
		\item when agents are either optimistic or pessimistic, and their preferences are separable and lexicographic, or 
		\item when agents' preferences are lexicographic and $O$-legal and the local mechanisms also satisfy non-bossiness.
	\end{enumerate}
\end{restatable}
\begin{proof}
	
	(1): Let $P$ be a profile of separable lexicographic preferences. Suppose for the sake of contradiction that an agent $j$ has a beneficial misreport $\succ'_j$, and let $P'=(P_{-j},\succ'_j)$. Let $k$ be the type of highest importance to $j$ for which $[f_O(P')(j)]_k \neq [f_O(P)(j)]_k$. Then, by our assumption that preferences are lexicographic, $k$ being the most important type for $j$ where her allocated item differs, and that $P'$ is a beneficial manipulation, it must hold that $[f_O(P')(j)]_k \succ [f_O(P)(j)]_k$. Since, preferences are separable,  $[f(P')]_k = f_k(\projs{P'}{\{D_k\}})$. Since every other agent is truthful, $\projs{P'}{\{D_k\}} = (\projs{P_{-j}}{\{D_k\}},\proj{\succ'_j}{D_k})$, and $\projs{\succ'_j}{D_k} \neq \projs{\succ_j}{D_k}$ is a beneficial manipulation, which implies that $f_k$ is not strategyproof, a contradiction to our assumption.
	
	(2) Now, we consider the case where the profile of truthful preferences $P$ is an arbitrary $O$-legal and lexicographic profile of preferences that may not be separable, and the local mechanisms are non-bossy and strategyproof. Suppose for the sake of contradiction that an agent $j$ has a beneficial misreport $\succ'_j$, and let $P'=(P_{-j},\succ'_j)$. W.l.o.g. let $O=[1\rhd\dots\rhd p]$. 
	
	Let $k$ be the most important type for which agent $j$ receives a different item. We begin by showing that by our assumption that the local mechanisms are non-bossy, and our assumption of $O$-legal lexicographic preferences, it holds that for every $i < k$ according to $O$, $\projs{f_i(P')}{D_i} = \projs{f_i(P)}{D_i}$. For the sake of contradiction, let $h <k$ be the first type for which some agent $l$ receives a different item, i.e. $[f(P')(l)]_{h} \neq [f(P)(l)]_{h}$, and $\projs{f(P')}{D_{<h}} = \projs{f(P)}{D_{<h}}$. 
	Then, by our assumption of $O$-legal lexicographic preferences, and every other agent reporting truthfully, $\proj{P'}{D_h}{\projs{f(P')}{D_{<h}}} = (\newproj{P}{-j}{D_h}{\projs{f(P')}{D_{<h}}}, \newproj{\succ'}{j}{D_h}{\projs{f(P')}{D_{<h}}})$. 
	By minimality of $k$, we know that $\projs{f_{h}(P')(j)}{D_h} = \projs{f_{h}(P)(j)}{D_h}$. But, $\projs{f_{h}(P')(l)}{D_h} \neq \projs{f_{h}(P)(l)}{D_h}$, which implies that $f_h$ does not satisfy non-bossiness, which is a contradiction.
	
    Now, by minimality of $k$ and our assumption that preferences are $O$-legal and lexicographic and that $k$ is the most important type for which any agents' allocation changes as we just showed, it must hold that $[f(P')(j)]_k = f_{k}(\proj{P'}{D_k}{\projs{f(P')}{D_{<k}}})(j) \succ f_{k}(\proj{P}{D_k}{\projs{f(P)}{D_{<k}}})(j) = [f(P)(j)]_k$. 
    However, $\projs{f(P')}{D_{<k}} = \projs{f(P)}{D_{<k}}$, and $\proj{P'_{-j}}{D_k}{\projs{f(P')}{D_{<k}}} = \proj{P_{-j}}{D_k}{\projs{f(P')}{D_{<k}}}$. This implies that $f_k$ is not strategyproof, which is a contradiction.
\end{proof}

Having established that it is possible to design strategyproof sequential mechanisms, we now turn our attention to strategyproof sequential mechanisms that satisfy other desirable properties such as non-bossiness, neutrality, monotonicity, and Pareto-optimality. In \Cref{thm:crnets}, we show that under $O$-legal preferences, a mechanism satisfies strategyproofness and non-bossiness of more important types if and only if it is an $O$-legal CR-net composed of mechanisms that satisfy the corresponding counterparts of these properties for allocating items of a single type, namely, local strategyproofness and non-bossiness.

\begin{restatable}{dfn}{dfncrnet}\label{dfn:crnet}{\em [CR-net]}
A (directed) conditional rule net (CR-net) $\mm$ over $\md$ is defined by \begin{enumerate}[wide,labelindent=0pt,label = (\roman*)] 
\item a directed graph $G=(\{D_1,...,D_p\},E)$, called the {\em dependency graph}, and 
\item for each $D_i$, there is a conditional rule table $\crt_i$ that contains a mechanism denoted $\proj{\mm}{D_i}{A}$ for $D_i$ for each allocation $A$ of all items of types that are parents of $D_i$ in $G$, denoted $\pa(D_i)$.\end{enumerate}
Let $O=[D_1\rhd\dots\rhd D_p]$, then a CR-net is $O$-legal if there is no edge $(D_i,D_l)$ in its dependency with $i>l$.
\end{restatable}
\begin{figure}
  \centering
  \includegraphics[width=\linewidth]{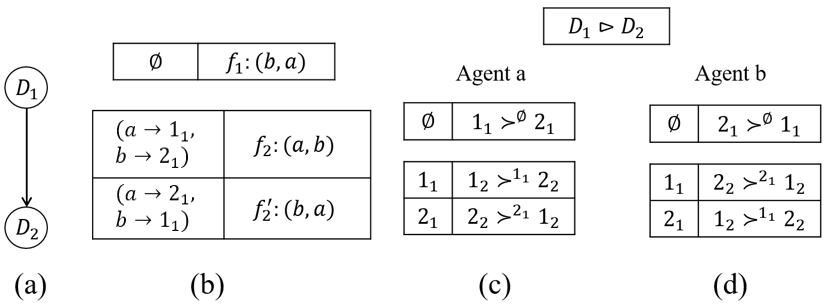}
  \caption{A serial dictatorship CR-net $f$.}
  \label{eg:crnet}
\end{figure}
\begin{example}
We note that the local mechanisms in a CR-net may be any mechanism that can allocate $n$ items to $n$ agents given strict preferences.
In Figure~\ref{eg:crnet}, we show a CR-net $f$ where all the local mechanisms are serial dictatorships.
The directed graph is shown in Figure~\ref{eg:crnet}(a), which implies $D_2$ depends on $D_1$. Figure~\ref{eg:crnet}(b) shows the CRT of $f$.
In the CRT, $f_1:(b,a)$ means that in the serial dictatorship $f_1$, agent $b$ picks her most preferred item first followed by agent $a$, and it is similar for $f_2,f'_2$.
The conditions in the CR-net, which are partial allocations are represented by mappings, for example, $(a\rightarrow 2_1)$ means agent $a$ gets $2_1$.
Figure~\ref{eg:crnet} (c) and (d) are the $O$-legal preferences of agents $a$ and $b$, respectively, where $O=[D_1\rhd D_2]$.
According to $f$, first we apply $f_1$ on $D_1$, and we have $a\rightarrow 1_1,b\rightarrow 2_1$.
Then, by CRT of $f$ we use $f_2$ for $D_2$, and we have $a\rightarrow 1_2,b\rightarrow 2_2$.
Therefore $f$ outputs an allocation where $a\rightarrow (1_1,1_2),b\rightarrow (2_1,2_2)$.
\end{example} 

\begin{restatable}{lem}{lemspsm}\label{lem:spsm}
	When agents' preferences are restricted to the $O$-legal lexicographic preference domain, for any strategyproof mechanism $f$, any profile $P$, and any pair $(P_{-j}, \succ'_j)$ obtained by agent $j$ misreporting her preferences by raising the rank of $f(P)(j)$ such that for any bundle $b$, $f(P)(j) \succ_j b \implies f(P)(j) \succ'_j b$, it holds that $f(P_{-j}, \succ'_j)(j) = f(P)(j)$.
\end{restatable}
\begin{proof}
	Suppose for the sake of contradiction that $f$ is a strategyproof mechanism that does not satisfy monotonicity. Let $P=(\succ_j)_{j\le n}$ be a profile, $j$ be an agent who misreports her preferences as $\succ'_j$ obtained from $\succ_j$ by raising the rank of $f(P)(j)$, specifically, for any bundle $b$, $f(P)(j) \succ_j b \implies f(P)(j) \succ'_j b$. Then, either: (1) $f(P_{-j},\succ'_j)(j) \succ'_j f(P)(j)$, or (2) $f(P)(j) \succ'_j f(P_{-j},\succ'_j)(j)$.
	
	(1) Suppose $f(P_{-j},\succ'_j)(j) \succ'_j f(P)(j)$. First, we claim that if $f(P_{-j},\succ'_j)(j)\allowbreak\succ'_j f(P)(j)$, then $f(P_{-j},\succ'_j)(j) \succ_j f(P)(j)$. Suppose for the sake of contradiction that this were not true, then $f(P)(j) \succ_j f(P_{-j},\succ'_j)(j)$ and $f(P_{-j},\succ'_j)(j) \succ'_j f(P)(j)$. This is a contradiction to our assumption on $\succ'_j$. This implies that $f(P_{-j},\succ'_j)(j) \succ_j f(P)(j)$ and $\succ'_j$ is a beneficial misreport for agent $j$, a contradiction to our assumption that $f$ is strategyproof.
	
	(2) If $f(P)(j) \succ'_j f(P_{-j},\succ'_j)(j)$, then $\succ_j$ is a beneficial misreport for agent $j$ w.r.t. $P'$, a contradiction to our assumption that $f$ is strategyproof.
\end{proof}

\begin{restatable}{thm}{thmcrnets}\label{thm:crnets}
	For any importance order $O$, a mechanism satisfies strategyproofness and non-bossiness of more important types under the $O$-legal lexicographic preference domain if and only if it is an $O$-legal locally strategyproof and non-bossy CR-net.
\end{restatable}
\begin{proof}
The if part is obvious (and is proved in Proposition~\ref{prop:localsp2sp}). We prove the only if part by induction. 

\begin{restatable}{claim}{cldecompose}\label{cl:decompose}
If an allocation mechanism satisfies non-bossiness of more important types and strategyproofness, then it can be decomposed into a locally strategyproof and non-bossy CR-net.
\end{restatable}
Proof by induction on the number of types. The claim is trivially true for the base case with $p=1$ type. Suppose the claim holds true for $p=k$ types i.e. when there are at most $k$ types, if an allocation mechanism is non-bossy in more important types and strategyproof, then it can be decomposed into locally strategyproof and non-bossy mechanisms.

When $p=k+1$, we prove that any non-bossy and strategyproof allocation mechanism $f$ for a basic type-wise allocation problem can be decomposed into two parts by Step 1:
\begin{enumerate}
\item Applying a local allocation mechanism $f_1$ to $D_1$ to compute allocation $[A]_1$.
\item Applying an allocation mechanism $\proj{f}{D_{-1}}{[A]_1}$ to types $D_{-1}$.
\end{enumerate}

\noindent $\bullet${\bf~Step 1.} For any strategyproof allocation mechanism satisfying non-bossiness of more important types, allocations for type $1$ depend only on preferences restricted to $D_1$.

\begin{restatable}{claim}{clpref1}\label{cl:pref1}
For any pair of profiles $P = (\succ_j)_{j\le n},Q = (\succ'_j)_{j\le n}$, and $\projs{P}{D_1} = \projs{Q}{D_1}$, we must have that $\projs{f(P)}{D_1} = \projs{f(Q)}{D_1}$.
\end{restatable}
\begin{proof}
Suppose for sake of contradiction that $\projs{f(P)}{D_1} \neq \projs{f(Q)}{D_1}$. For any $0 \le j \le n$, define $P_j = (\succ'_1,\dots,\succ'_j,\succ_{j+1},\dots,\succ_n)$ and suppose $\projs{f(P_j)}{D_1} \neq \projs{f(P_{j+1})}{D_1}$ for some $j \le n-1$. Let $[A]_1 = \projs{f(P_j)(j+1)}{D_1}$ and $[B]_1 = \projs{f(P_{j+1})(j+1)}{D_1}$. Now, suppose that 

Case 1: $[A]_1 = [B]_1$, but for some other agent $\hat j$, $\projs{f(P_j)(\hat j)}{D_1} \neq \projs{f(P_{j+1})(\hat j)}{D_1}$. This is a direct violation of non-bossiness of more important types because $\newprojs{P}{j}{D_1} = \newprojs{P}{j+1}{D_1}$ by construction.

Case 2: $[A]_1 \neq [B]_1$. If $[B]_1 \newprojs{\succ}{j+1}{D_1} [A]_1$, then $(P_j,\succ'_{j+1})$ is a beneficial manipulation due to agents' lexicographic preferences. Otherwise, if $[A]_1 \newprojs{\succ}{j+1}{D_1} [B]_1$, then $(P_{j+1},\succ_{j+1})$ is a beneficial manipulation due to our assumption that $\newprojs{\succ}{j+1}{D_1} = \newprojs{\succ'}{j+1}{D_1}$ and agents' lexicographic preferences. This contradicts the strategyproofness of $f$.
\end{proof}

\noindent $\bullet${\bf~Step 2.} Show that $f_1$ satisfies strategyproofness and non-bossiness.

First, we show that $f_1$ must satisfy strategyproofness by contradiction. Suppose for the sake of contradiction that $f$ is strategyproof but $f_1$ is not strategyproof. Let $P=(\succ_j)_{j \le n}$ be a profile of agents' preferences over $D_1$. Then, there exists an agent $j^*$ with a beneficial manipulation $\succ'_{j^*}$. Now, consider a profile $Q = (\bar\succ_j)_{j \le n}$ where for every agent $j, \newprojs{\bar\succ}{j}{D_1} = \succ_j$ and the mechanism $f$ whose local mechanism for $D_1$ is $f_1$. We know from Step 1 that $\projs{f(Q)}{D_1} = f_1(\projs{Q}{D_1}) = f_1(P)$. However, in that case, because agents' preferences are lexicographic with $D_1$ being the most important type, agent $j^*$ has a successful manipulation $\bar\succ'_{j^*}$ where $\newprojs{\bar\succ'}{j^*}{D_1} = \succ'_{j^*}$ since the resulting allocation of $f_1(\bar\succ_{-j^*},\bar\succ'_{j^*})$ is a strictly preferred item of type $D_1$. This is a contradiction to our assumption on the strategyproofness of $f$.

Then, we also show that $f_1$ satisfies non-bossiness. 
Suppose for the sake of contradiction that $f_1$ is not non-bossy. Let $P=(\succ_j)_{j \le n}$ be a profile of agents' preferences over $D_1$. Then, there exists an agent $j^*$ with a bossy preference $\succ'_{j^*}$ such that for $P'=(\succ_{-j^*},\succ'_{j^*})$, $f_1(P)(j^*)=f_1(P')(j^*)$ while $f_1(P)(j)\neq f_1(P')(j)$ for some $j$. Now, consider a profile $Q = (\bar\succ_j)_{j \le n}$ where for every agent $j, \newprojs{\bar\succ}{j}{D_1} = \succ_j$ and the mechanism $f$ whose local mechanism for $D_1$ is $f_1$. We know from Step 1 that $\projs{f(Q)}{D_1} = f_1(\projs{Q}{D_1}) = f_1(P)$. However, in that case, because agents' preferences are lexicographic with $D_1$ being the most important type, agent $j^*$ has a bossy preference $\bar\succ'_{j^*}$ where $\newprojs{\bar\succ'}{j^*}{D_1} = \succ'_{j^*}$ such that $\projs{f(Q)(j^*)}{D_1}=\projs{f(\bar\succ_{-j^*},\bar\succ'_{j^*})(j^*)}{D_1}$ while $\projs{f(Q)(j)}{D_1}\neq \projs{f(\bar\succ_{-j^*},\bar\succ'_{j^*})(j)}{D_1}$ for some $j$. This is a contradiction to our assumption that $f$ satisfies non-bossiness of more important types.

\noindent $\bullet${\bf~Step 3.} The allocations for the remaining types only depend on the allocations for $D_1$.
\begin{restatable}{claim}{cldep1}\label{cl:dep1}
Consider any pair of profiles $P_1,P_2$ such that $[A]_1 = f_1(\newprojs{P}{1}{D_1}) = f_1(\newprojs{P}{2}{D_1})$, and $\newproj{P}{1}{D_{-1}}{[A]_1} = \newproj{P}{2}{D_{-1}}{[A]_1}$, then $f(P_1) = f(P_2)$.
\end{restatable}
\begin{proof}
We prove the claim by constructing a profile $P$ such that $f(P) = f(P_1) = f(P_2)$.

Let $P_1 = (\succ_j)_{j \le n}$, $P_2 = (\bar\succ_j)_{j\le n}$ and $P=(\hat\succ_j)_{j\le n}$. Let $\hat\succ_j$ be obtained from $\succ_j$ by changing the preferences over $D_1$ by raising $[A]_1(j)$ to the top position. Agents' preference over $D_{-1}$ are $\newproj{\hat\succ}{j}{D_{-1}}{[A]_1} = \newproj{\succ}{j}{D_{-1}}{[A]_1} ( = \newproj{\bar\succ}{j}{D_{-1}}{[A]_1})$. It is easy to check that for every bundle $b$, $f(P)(j) \succ_j b \implies f(P)(j) \hat\succ_j b$.
By applying Lemma~\ref{lem:spsm} sequentially to every agent, $f(P) = f(P_1)$.
Similarly, $f(P) = f(P_2)$. It follows that for any allocation $[A]_1$ of items of type $D_1$, there exists a mechanism $\proj{f}{D_{-1}}{[A]_1}$ such that for any profile $P$, we can write $f(P)$ as $(f_1(\projs{P}{D_1}), \proj{f}{D_{-1}}{[A]_1}(\proj{P}{D_{-1}}{[A]_1}))$.
\end{proof}
\noindent $\bullet${\bf~Step 4.} Show that $\proj{f}{D_{-1}}{[A]_1}$ satisfies strategyproofness and non-bossiness of important types for any allocation $[A]_1$ of $D_1$.

Suppose for the sake of contradiction that $\proj{f}{D_{-1}}{[A]_1}$ is not strategyproof for some profile $\proj{P}{D_{-1}}{[A]_1}$. Then, for $P=(\succ_j)_{j \le n}$ there is an agent $j^*$ with a beneficial manipulation w.r.t. $P$ and $[A]_1$, $\newproj{\succ'}{j^*}{D_{-1}}{[A]_1} \neq \newproj{\succ}{j^*}{D_{-1}}{[A]_1}$ and $\newprojs{\succ'}{j^*}{D_1} = \newprojs{\succ}{j^*}{D_1}$. Let $Q = (\succ_{-j^*},\succ'_{j^*})$. Then, $f(Q)(j) = ([A]_1,\proj{f}{D_{-1}}{[A]_1}(\proj{Q}{D_{-1}}{[A]_1}))(j) \succ_j ([A]_1,\proj{f}{D_{-1}}{[A]_1}(\proj{P}{D_{-1}}{[A]_1}))(j) = f(P)(j)$. This is a contradiction to the strategyproofness of $f$.

Suppose for sake of contradiction that $\proj{f}{D_{-1}}{[A]_1}$ does not satisfy non-bossiness of important types. Then, there is a profile $P = (\succ_j)_{j \le n}$, and an agent $j^*$ with a bossy manipulation of her preferences $\newproj{\succ}{j^*}{D_{-1}}{[A]_1}$. Then, it is easy to verify that $f$ also does not satisfy non-bossiness of important types.

In Step 1, we showed that the allocation for $D_1$ only depends on the restriction of agents' preferences to $D_1$ i.e. over $\projs{P}{D_1}$. In Step 3 we showed that $f(P)$ can be decomposed as $(f_1(\projs{P}{D_1}),\proj{f}{D_{-1}}{[A]_1}(\proj{P}{D_{-1}}{[A]_1}))$ where $[A]_1 = f_1(\projs{P}{D_1})$. In Steps 2 we showed that $f_1$ must be strategyproof and non-bossy. In Step 4, we showed that for any output $[A]_1$ of $f_1$, the mechanism $\proj{f}{D_{-1}}{[A]_1}$ satisfies both strategyproofness and non-bossiness of important types i.e. that we can apply the induction assumption that $\proj{f}{D_{-1}}{[A]_1}$ is a locally strategyproof and non-bossy CR-net of allocation mechanisms. Together with the statement of Step 2, this completes the inductive argument. 
\end{proof}

In \Cref{thm:char}, we characterize the class of strategyproof, non-bossy of more important types, and type-wise neutral mechanisms under  $O$-legal lexicographic preferences, as the class of $O$-legal sequential compositions of serial dictatorships. The proof relies on \Cref{thm:crnets} and \Cref{cl:nCR}, where we show that any CR-net mechanism that satisfies type-wise neutrality is an $O$-legal sequential composition of neutral mechanisms, one for each type.

\begin{restatable}{claim}{neutralCR-net}\label{cl:nCR}
For any importance order $O$, an $O$-legal CR-net with type-wise neutrality is an $O$-legal sequential composition of neutral mechanisms.
\end{restatable}
\begin{proof}
We prove the claim by induction. Suppose $f$ is such a CR-net.
From the decomposition in the proof of Claim~\ref{cl:decompose}, we observe that the mechanism used for type $i$ depends on \projs{f(P)}{D_{\le i}}.
From this observation, and the importance order $O$, we can deduce that the mechanism for type $1$ depends on no other type, and therefore there is only one mechanism for type $1$, say, $f_1$. First we show that $f_1$ is neutral.
Otherwise, there exists a permutation $\Pi_1$ over $D_1$, $f_1(\Pi_1(\projs{P}{D_1}))\neq \Pi_1(f_1(\projs{P}{D_1}))$.
Let $I=(I_i)_{i\le p}$ where $I_i$ is the identity permutation for type $i$.
Then for $\Pi=(\Pi_1,I_{-1})$, we have $\projs{f(\Pi(P))}{D_1}=f_1(\Pi_1(\projs{P}{D_1}))\neq \Pi_1(f_1(\projs{P}{D_1}))=\projs{\Pi(f(P))}{D_1}$, a contradiction.

Now, suppose that for a given $i$, there is only one mechanism $f_{i'}$ for each type $i'\le i$, and each $f_{i'}$ is neutral.
Let $\Pi=(\Pi_{\le i},I_{> i})$ and we have $\projs{f(\Pi(P))}{D_{\le i}}=\projs{\Pi(f(P))}{D_{\le i}}$.
Let $A=\projs{f(P)}{D_{\le i}}$ and $B=\projs{f(\Pi(P))}{D_{\le i}}=\Pi_{\le i}(A)$.
Because $P$ is chosen arbitrarily, $A$ and $B$ are also arbitrary outputs of mechanism $f$ over $D_{\le i}$.
Let $f_{i+1}=\proj{f}{D_{i+1}}{A}$,and $f'_{i+1}=\proj{f}{D_{i+1}}{B}$.
Similarly both $f_{i+1}$ and $f'_{i+1}$ are arbitrary mechanisms in $CRT$.
Because $f$ is neutral, we have $\projs{f(\Pi(P))}{D_{i+1}}=\projs{\Pi(f(P))}{D_{i+1}}$, i.e. $f_{i+1}(\proj{P}{D_{i+1}}{A})= f'_{i+1}(\proj{M(P)}{D_{i+1}}{B})$.
By assumption we know that $\Pi_{i+1}=I_{i+1}$, so $\proj{P}{D_{i+1}}{A}=\proj{\Pi(P)}{D_{i+1}}{B}$.
That means $f_{i+1}$ and $f'_{i+1}$ can replace each other in $CRT$ of $f$ for type $i+1$.
Therefore in fact there is only one mechanism $f_{i+1}$ for type $i+1$ in $CRT$.

Moreover $f_{i+1}$ must be neutral.
Otherwise, there must be some permutation $\Pi_{i+1}$ over $D_{i+1}$, $f_{i+1}(\Pi_{i+1}(\proj{P}{D_{i+1}}{A}))\neq \Pi_{i+1}(f_{i+1}(\proj{P}{D_{i+1}}{A}))$.
Then for $\Pi=(\Pi_{\le i+1},I_{> i+1})$, we have $\projs{f(\Pi(P))}{D_{i+1}}=f_{i+1}(\proj{\Pi(P)}{D_{i+1}}{B})=f_{i+1}(\Pi_{i+1}(\proj{P}{D_{i+1}}{A}))\neq \Pi_{i+1}(f_{i+1}(\proj{P}{D_{i+1}}{A}))=\projs{\Pi(f(P))}{D_{i+1}}$, a contradiction.
\end{proof}

\begin{restatable}{thm}{characterize}\label{thm:char}
For any importance order $O$, under the $O$-legal lexicographic preference domain, an allocation mechanism satisfies strategyproofness, non-bossiness of more important types, and type-wise neutrality if and only if it is an $O$-legal sequential composition of serial dictatorships.
\end{restatable}
\begin{proof}
Let $O=[D_1\rhd D_2\rhd\dots\rhd D_p]$.
When $p=1$, we know that serial dictatorship is characterized by strategyproofness, non-bossiness, and neutrality~\cite{Mackin2016:Allocating}.
Let $P=(\succ_j)_{j\le n}$ be an arbitrary $O$-legal lexicographic preference profile.

\noindent{\bf$\Rightarrow$}: Let $f_O=(f_1,\dots,f_p)$. It follows from~\Cref{thm:crnets} that if each $f_i$ satisfies strategyproofness and non-bossiness, then $f_O$ satisfies strategyproofness and non-bossiness of more important types, because $f_O$ can be regarded as a CR-net with no dependency among types.
If each $f_i$ satisfies neutrality, then by~\Cref{thm:localtoglobal} we have that $f$ satisfies type-wise neutrality.
Therefore, since each $f_i$ is a serial dictatorship, which implies that it satisfies strategyproofness, non-bossiness, and neutrality, we have that $f_O$ satisfies strategyproofness, non-bossiness of more important types, and type-wise neutrality.

\noindent{\bf$\Leftarrow$}:
We now prove the converse. Let $f$ be a strategyproof and non-bossy mechanism under $O$-legal lexicographic preferences. Then by~\Cref{thm:crnets}, we have that $f$ is an $O$-legal strategyproof and non-bossy CR-net. The rest of the proof depends on the following claim:

\Cref{cl:nCR} implies that there is only one mechanism $f_i$ for each type $i$ in $CRT$, and $f_i$ is neutral.
Therefore with Theorem~\ref{thm:crnets} and Claim~\ref{cl:nCR}, if $f$ satisfies strategyproofness, non-bossiness of more important types, and type-wise neutrality, we have that $f$ is an $O$-legal sequential composition of local mechanisms that are strategyproof, non-bossy, and neutral, which implies that they are serial dictatorships~\cite{Mackin2016:Allocating}.
\end{proof}

\begin{restatable}{thm}{thmcharcrnet}\label{thm:char2}
For any arbitrary importance order $O$, under the $O$-legal lexicographic preference domain, an allocation mechanism satisfies strategyproofness, non-bossiness of more important types, and Pareto-optimality if and only if it is an $O$-legal CR-net composed of serial dictatorships.
\end{restatable}
\begin{proof}
(Sketch)
For a single type, we know that serial dictatorship is characterized by strategyproofness, Pareto-optimality, and non-bossiness~\cite{Papai01:Strategyproof}.
The proof is similar to~\Cref{thm:char}, and uses a similar argument to Theorems~\ref{thm:localtoglobal} and~\ref{thm:globaltolocal}, to show that an $O$-legal CR-net is Pareto-optimal if and only if every local mechanism is Pareto-optimal. The details are provided in the apoendix.
\end{proof}

Finally, we revisit the question of strategyproofness when preferences are not $O$-legal w.r.t. a common importance order. We show in~\Cref{thm:manipulationhard} that even when agents' preferences are restricted to lexicographic preferences, there is a computational barrier against manipulation; determining whether there exists a beneficial manipulation w.r.t. a sequential mechanism is NP-complete for \mtra{}s, even when agents' preferences are lexicographic. Details and the full proof are relegated to the appendix.

\begin{restatable}{dfn}{dfnbm}\label{dfn:bm}
Given an \mtra{} $(N,M,P)$, where $P$ is a profile of lexicographic preferences, and a sequential mechanism $f_O$. in {\sc BeneficialManipulation}, we are asked whether there exists an agent $j$ and an $O$-legal lexicographic preference relation $\succ'_j$ such that $f_O((P_{-j},\succ'_j))(j) \succ_j f_O(P)(j)$.
\end{restatable}

\begin{restatable}{thm}{thmmanipulationhard}\label{thm:manipulationhard}
	{\sc BeneficialManipulation} is NP-complete when preferences are not $O$-legal.
\end{restatable}

\section{Conclusion and Future Work}
We studied the design of strategyproof sequential mechanisms for \mtra{}s under $O$-legal lexicographic preferences, and showed the relationship between properties of sequential mechanisms and the local mechanisms that they are composed of. In doing so, we obtained strong characterization results showing that any mechanism satisfying strategyproofness, and combinations of appropriate notions of non-bossiness, neutrality, and Pareto-optimality for \mtra{}s must be a sequential composition of local mechanisms. This decomposability of strategyproof mechanisms for \mtra{}s provides a fresh hope for the design of decentralized mechanisms for \mtra{}s and multiple assignment problems. Going forward, there are several interesting open questions such as whether it is possible to design decentralized mechanisms for \mtra{}s that are fair, efficient, and strategyproof under different preference domains. 

\section*{Acknowledgments}
LX acknowledges support from NSF \#1453542 and \#1716333 and a gift fund from Google. YC acknowledges NSFC under Grants 61772035 and 61932001.

\bibliographystyle{plainnat}
\bibliography{references,ref}

\clearpage

\section{Appendix}

\subsection{Proof of~\Cref{thm:localtoglobal}}
\thmlocaltoglobal*
\begin{proof}
	Throughout, we will assume that $O=[D_1\rhd\dots\rhd D_p]$, and that $P$ is an arbitrary $O$-legal preference profile over $p$ types. For any $i\le p$, we define $g_i$ to be the sequential mechanism $(f_1,\dots,f_i)$.
	
	\noindent{\bf anonymity}. It is easy to see that the claim is true when $p=1$. Now, suppose that claim is true for all $p\le k$. Let $P$ be an arbitrary profile over $k+1$ types.
	Let $g=(f_2,\dots,f_{k+1})$. Now, 
	$\Pi(f_O(P)) = (\Pi(f_1(\projs{P}{D_1})),\Pi(g(\proj{P}{D_{\le k+1}\setminus D_1}{\Pi(f_1(\projs{P}{D_1}))})))$, and 
	$f_O(\Pi(P)) = (f_1(\Pi(\projs{P}{D_1})),g(\Pi(\proj{P}{D_{\le k+1}\setminus D_1}{f_1(\Pi(\projs{P}{D_1}))})))$. 
	Since $f_1$ is anonymous, $\Pi(f_1(\projs{P}{D_1})) = f_1(\Pi(\projs{P}{D_1}))$. 
	Therefore, $\proj{P}{D_{\le k+1}\setminus D_1}{\Pi(f_1(\projs{P}{D_1}))} = \proj{P}{D_{\le k+1}\setminus D_1}{f_1(\Pi(\projs{P}{D_1}))}$.
	Then, by the induction assumption, $g$ satisfies anonymity, and we have  $\Pi(g(\proj{P}{D_{\le k+1}\setminus D_1}{\Pi(f_1(\projs{P}{D_1}))})) = g(\Pi(\proj{P}{D_{\le k+1}\setminus D_1}{f_1(\Pi(\projs{P}{D_1}))}))$.
	It follows that $\Pi(f_O(P))= f_O(\Pi(P))$.
	
    \noindent{\bf type-wise neutrality.} We show only the induction step. Suppose that the claim is always true when $p\le k$. Let $P$ be an arbitrary profile over $k+1$ types. Let $g=(f_2,\dots,f_{k+1})$, and $\Pi_{-1}=(\Pi_2,\dots,\Pi_{k+1})$.
    Let $A_1=\Pi_1(f_O(\projs{P}{D_1}))$ and $B_1=f_1(\Pi_1(\projs{P}{D_1}))$.
    Now,
	$\Pi(f_O(P)) = (A_1, \Pi_{-1}(g(\proj{P}{{D_{\le k+1}\setminus D_1}}{A_1})))$, and 
	$f_O(\Pi(P)) = (B_1, g(\Pi_{-1}(\proj{P}{D_{\le k+1}\setminus D_1}{B_1})))$.
	Since $f_1$ is neutral, $A_1 = B_1$.
	Then, $\proj{P}{D_{\le k+1}\setminus D_1}{A_1} = \proj{P}{D_{\le k+1}\setminus D_1}{B_1}$.
	Then, by the induction assumption, $g$ satisfies type-wise neutrality, and $\Pi_{-1}(g(\proj{P}{D_{\le k+1}\setminus D_1}{A_1})) = g(\Pi_{-1}(\proj{P}{D_{\le k+1}\setminus D_1}{B_1}))$.
	It follows that $\Pi(f_O(P)) = f_O(\Pi(P))$.
	
	\noindent{\bf non-bossiness.}
	Let us assume for the sake of contradiction that the claim is false, i.e. there exists a profile $P$, an agent $j$ and a misreport $\succ'_j$ such that for $P' = (\succ_{-j},\succ'_{j})$, $f_O(P)(j) = f_O(P')(j)$, and $f_O(P) \neq f_O(P')$. Then, there is a type $i\le p$ such that, $\projs{f_O(P)}{D_{<i}} = \projs{f_O(P')}{D_{<i}}$ and $\projs{f_O(P)}{D_i} \neq \projs{f_O(P')}{D_i}$. Let $A = \projs{f_O(P)}{D_{<i}}$. Then, there is an agent $k$ such that $f_{i}(\proj{P}{D_i}{A})(k)\neq f_{i}(\proj{P'}{D_i}{A})(k)$. 
	By the choice of $i$, and the assumption that every other agent reports preferences truthfully, $\proj{\succ_j}{D_i}{A}\neq \proj{\succ'_j}{D_i}{A}$. Then, $f_i(\proj{\succ_{-j}}{D_i}{A},\proj{\succ_j}{D_i}{A})(j)= f_i(\proj{\succ_{-j}}{D_i}{A},\proj{\succ'_j}{D_i}{A})(j)$, but $f_i(\proj{\succ_{-j}}{D_i}{A},\proj{\succ_j}{D_i}{A})(k)\neq f_i(\proj{\succ_{-j}}{D_i}{A},\proj{\succ'_j}{D_i}{A})(k)$, a contradiction to our assumption that $f_i$ is non-bossy.

	\noindent{\bf monotonicity.} 
	Let $P'=(P_{-j},\succ'_j)$ be an $O$-legal profile obtained from $P$ and $Y\subseteq\md$ is the set of bundles raising the ranks in $P'$ such that the relative rankings of bundles in $Y$ are unchanged in $P$ and $P'$. 
	For any $Y\subseteq\md$, and any $\vu\in\md_{D_{<i}}$, let $Y^{D_i\mid \vu} = \{x_i: \vx\in Y,x_h=u_h \text{ for all } h\le i-1\}$. 
	It is easy to see that if $\vx_1 = \projs{f_O(P')(j)}{\{D_1\}}$, then it follows from strong monotonicity of $f_1$ that $\vx_1\in\projs{f_O(P)(j)}{\{D_1\}}\cup Y^{D_1}$.
	Now, either $\vx_1\neq\projs{f_O(P)(j)}{\{D_1\}}$, or $\vx_1=\projs{f_O(P)(j)}{\{D_1\}}$. 
	Suppose $\vx_1\neq\projs{f_O(P)(j)}{\{D_1\}}$. 
	Then, by strong monotonicity of $f_1$, $\vx_1\succ\projs{f_O(P)(j)}{\{D_1\}}$. 
	Then, by our assumption of $O$-legal lexicographic preferences, for any $\vz\in\md_{\{D_2,\dots,D_p\}}$, $(\vx_1,\vz)\in Y$. 
	Therefore, $f_O(P')(j)\in Y$. Suppose $\vx_1=\projs{f_O(P)(j)}{\{D_1\}}$, then by a similar argument, $\projs{f_O(P')(j)}{\{D_2\}}\in \{\projs{f_O(P)(j)}{\{D_2\}}\}\cup Y^{D_2\mid (\vx_1)}$. 
	Applying our argument recursively, we get that $f_O(P')(j)\in \{f_O(P)(j)\}\cup Y$.
	
	\noindent{\bf Pareto-optimality.}
	Suppose the claim is true for $p\le k$ types. Let $P$ be an $O$-legal lexicographic profile over $k+1$ types, and $f_O=(f_i)_{i\le k+1}$ is a sequential composition of Pareto-optimal local mechanisms. Suppose for the sake of contradiction that there exists an allocation $B$ such that some agents strictly better off compared to $f_O(P)$, and no agent is worse off. Then, by our assumption of lexicographic preferences, for every agent $k$ who is not strictly better off, $B(k) = f_O(P)(k)$, and for every agent $j$ who is strictly better off, one of two cases must hold. (1) $\projs{B(j)}{D_1}\succ_j \projs{f_O(P)(j)}{D_1}$, or (2) $\projs{B(j)}{D_1} = \projs{f_O(P)(j)}{D_1}$. (1): If there exists an agent such that $\projs{B(j)}{D_1}\succ_j \projs{f_O(P)(j)}{D_1}$, this is a contradiction to our assumption that $f_1$ is Pareto-optimal. (2): Suppose $\projs{B(j)}{D_1}=\projs{f_O(P)(j)}{D_1}$ for all agents who are strictly better off. Let $g=(f_2,\dots,f_{k+1})$. W.l.o.g. let agent $1$ strictly prefer $B(1)$ to $f_O(P)(1)$. Then, $g(\proj{P}{D_{\le k+1}\setminus D_1}{\projs{f_O(P)}{D_1}})(1)\succ_1 \projs{B(1)}{D_{\le k+1}\setminus D_1}$, and for every other agent $l\neq 1$, either $g(\proj{P}{D_{\le k+1}\setminus D_1}{\projs{f_O(P)}{D_1}})(l) \succ_l \projs{B(l)}{D_{\le k+1}\setminus D_1}$, or $g(\proj{P}{D_{\le k+1}\setminus D_1}{\projs{f_O(P)}{D_1}})(l) =  \projs{B(l)}{D_{\le k+1}\setminus D_1}$, which is a contradiction to our induction assumption.
\end{proof}

\subsection{Proof of~\Cref{thm:globaltolocal}}
\thmglobaltolocal*
\begin{proof}
	\noindent{\bf anonymity.} 
	Suppose that for some $k\le p$, $f_k$ does not satisfy anonymity. Then, there exists a profile $P_k$ on $D_k$ such that for some permutation $\Pi$ on agents $f_k(\Pi(P_k) \neq \Pi(f_k(P_k))$. Now, consider the $O$-legal separable lexicographic profile $P$, where for any type $i\le p$, the preferences over type $D_i$ is denoted $\projs{P}{D_i}$ and $\projs{P}{D_k}=P_k$. It is easy to see that, $f_O(\Pi(P)) = (f_i(\Pi(\projs{P}{D_i})))_{i\le p}$, and  
	$\Pi(f_O(P))= \Pi(f_1(\projs{P}{D_1}),\dots,f_p(\projs{P}{D_p})) = (\Pi(f_1(\projs{P}{D_1})),\dots,\Pi(f_p(\projs{P}{D_p}))))$. By anonymity of $f$, $f_O(\Pi(P)) = \Pi(f_O(P))$, which implies that $f_k(\Pi(\projs{P}{D_k})) = \Pi(f_k(\projs{P}{D_k}))$, which is a contradiction.
	
	\noindent{\bf type-wise neutrality.} 
	Suppose that some $k\le p$, $f_k$ does not satisfy neutrality. Then, there exists a profile $P_k$ on $D_k$ such that for some permutation $\Pi_k$ on $D_k$ $f_k(\Pi_k(P_k) \neq \Pi_k(f_k(P_k))$. 
	Now, consider the $O$-legal separable lexicographic profile $P$, where for any type $i\le p$, the preferences over type $D_i$ is denoted $\projs{P}{D_i}$ and $\projs{P}{D_k}=P_k$, and let $\Pi = (\Pi_1,\dots,\Pi_k,\dots,\Pi_p)$ be a permutation over $\md$ by applying $Pi_i$ on $D_i$ for each type $i\le p$.
	$f_O(\Pi(P)) = (f_i(\Pi_i(\projs{P}{D_i})))_{i\le p}$, and $\Pi(f_O(P)) =  (\Pi_1(f_1(\projs{P}{D_1})),\dots,\Pi_p(f_p(\projs{P}{D_p}))))$. By type-wise neutrality of $f_O$, $f_O(\Pi(P)) = \Pi(f_O(P))$. This implies that $f_k(\Pi_k(\projs{P}{D_k})) = \Pi_k(f_k(\projs{P}{D_k}))$, where $\projs{P}{D_k})=P_k$, which is a contradiction.
	
	\noindent{\bf non-bossiness.}
	Assume for the sake of contradiction that $k\le p$ is the most important type such that $f_{k}$ does not satisfy non-bossiness. Then, there exists a preference profile $Q = (\succ^k)_{j \le n}$ over $D_{k}$, and a bossy agent $l$ and a misreport $Q' = (\succ^k_{-l}, \bar\succ^k_{l})$, such that  $f_k(Q')(l) = f_k(Q)(l)$, but $f_{k}(Q') \neq f_{k}(Q)$. Now, consider the $O$-legal separable lexicographic profile $P$, where for any type $i\le p$, the preferences over type $D_i$ is denoted $\projs{P}{D_i}$ and $\projs{P}{D_k}=Q$, and the profile $P'$ obtained from $P$ by replacing $\succ_l$ with $\succ'_l$, which in turn is obtained from $\succ_l$ by replacing $\projs{\succ_l}{D_k}$ with $\bar\succ^k_l$. It is easy to see that $\projs{f_O(P')}{D_{<k}}=\projs{f_O(P)}{D_{<k}}$, and $\projs{f_O(P')(l)}{\{D_k\}} = \projs{f_O(P)(l)}{\{D_k\}}$, but $\projs{f_O(P')}{\{D_k\}} \neq \projs{f_O(P)}{\{D_k\}}$, and by our assumption of separable preferences, $\projs{f_O(P')}{D_{>k}}=\projs{f_O(P)}{D_{>k}}$. This implies that $f_O(P')(l)=f_O(P)(l)$, but $f_O(P')\neq f_O(P)$, implying that $f_O$ does not satisfy non-bossiness, which is a contradiction.
	
	\noindent{\bf monotonicity.} Suppose for the sake of contradiction that $k$ is the most important type for which $f_k$ does not satisfy monotonicity. Then, there exists a profile $Q=(\succ^k_j)_{j\le n}$ of linear orders over $D_k$, such that for some agent $j$, $\bar\succ^k_l$ obtained from $\succ^k_l$ by raising the rank of a set of items $Z\subseteq D_k$ without changing their relative order, $f_k((Q{-l},\bar\succ_l))(l)\not\in \{f_k(Q)(l)\}\cup Z$. Now, consider the $O$-legal separable lexicographic profile $P$, where for any type $i\le p$, the preferences over type $D_i$ is denoted $\projs{P}{D_i}$ and $\projs{P}{D_k}=Q$, and the profile $P'$ obtained from $P$ by replacing $\succ_l$ with $\succ'_l$, which in turn is obtained from $\succ_l$ by replacing $\projs{\succ_l}{D_k}$ with $\bar\succ^k_l$. It is easy to see that $\projs{f_O(P')}{D_{<k}}=\projs{f_O(P)}{D_{<k}}$, and $\projs{f_O(P')(l)}{D_k}\notin \projs{f_O(P)(l)}{D_k}\cup Z$. By our assumption of $O$-legal separable lexicographic preferences, this implies that $f_O$ does not satisfy monotonicity, which is a contradiction.
	
    \noindent{\bf Pareto-optimality.} 
	Suppose that some $k\le p$, $f_k$ does not satisfy Pareto-optimality. Then, there exists a profile $P_k$ such that $f_k(P_k)$ is Pareto-dominated by an allocation $B$ of $D_k$. Now, consider the $O$-legal separable lexicographic profile $P$, where for any type $i\le p$, the preferences over type $D_i$ is denoted $\projs{P}{D_i}$ and $\projs{P}{D_k}=P_k$. Then, $f_O(P) = (f_i(\proj{P}{D_i}))_{i\le p}$ is Pareto-dominated by the allocation $B$ of all types, where for all types $i\neq k$, $\projs{B}{D_i} = f_i(\projs{P}{D_i})$, and $\projs{B}{D_k} = A$, which is a contradiction to the assumption that $f_O$ is Pareto-optimal.
\end{proof}

\subsection{Proof of~\Cref{thm:char2}}
\thmcharcrnet*
\begin{proof}
Let $O=[D_1\rhd D_2\rhd\cdots\rhd D_p]$.
Under single type, we know that serial dictatorship is characterized by strategyproofness, Pareto-Optimality, and non-bossiness~\cite{Papai01:Strategyproof}.
Let $P=(\succ_j)_{j\le n}$ be an arbitrary $O$-legal lexicographic preference profile.

\noindent{$\Rightarrow$}: Let $f$ be an $O$-legal CR-net. From~\Cref{thm:crnets} we know that if each local mechanism of $f$ satisfies strategyproofness and non-bossiness, then $f$ satisfies strategyproofness and non-bossiness of more important types.

We now prove that if each local mechanism is Pareto-Optimal, then $f$ is Pareto-optimal, similarly to~\Cref{thm:localtoglobal}.
Suppose for the sake of contradiction that $f$ is not Pareto-optimal, i.e. for some $P$, the allocation $B=(B_i)_{i\le p}$ Pareto-dominates $f(P)=A=(A_i)_{i\le p}$.
Let $i$ be the most important type that $A$ and $B$ and different allocuation, and we have $A_{<i}=B_{<i}$ and $B_i$ Pareto-dominates $A_i$.
Let $P_i=\proj{P}{D_i}{A_{<i}}$.
However, by the assumption that $f$ is a CR-net, we know that $A_i=\proj{f}{D_i}{A_{<i}}(P_i)$ is Pareto-optimal, i.e. $A_i$ does not Pareto-dominated by $B_i$, which is a contradiction.

Therefore if each local mechanism of $f$ is a serial dictatorship, which implies that it satisfies strategyproofness, non-bossiness, and Pareto-optimality, then $f$ satisfies strategyproofness, non-bossiness of more important types, and Pareto-optimality.

\noindent{$\Leftarrow$}:
Let $f$ be a mechanism for $O$-legal lexicographic preferences.
With Theorem~\ref{thm:crnets}, we have that if $f$ satisfies strategyproofness and non-bossiness of more important types, then it is an $O$-legal strategyproof and non-bossy CR-net.
We also have that if $f$ is a CR-net satisfying Pareto-optimality, then each local mechanism is also Pareto-optimal with a similar proof to Theorem~\ref{thm:globaltolocal}.
Together we have that if $f$ satisfies strategyproofness, non-bossiness of more important types, and Pareto-optimality, then $f$ is an $O$-legal CR-net and each local mechanism satisfies strategyproofness, non-bossiness, and Pareto-Optimality, which implies that it is a serial dictatorship~\cite{Papai01:Strategyproof}.
%
%
\end{proof}

\subsection{Proof of~\Cref{thm:manipulationhard}}
\thmmanipulationhard*
\begin{proof}
	We show a reduction from 3-SAT. In an instance $I$ of 3-SAT involving $s$ Boolean variables $\{x_1,\dots,x_s\}$, and a formula $\mf$ involving $t$ clauses $\{c_1,\dots,c_t\}$ in 3-CNF, we are asked if 
	$\mf$ is satisfiable. Given such an arbitrary instance $I$ of 3-SAT, we construct an instance $J$ of {\sc BeneficialManipulation} in polynomial time. We will show that $I$ is a Yes instance of 3-SAT if and only if $J$ is a Yes instance of {\sc BeneficialManipulation}. For each $j\le t$, we label the three literals in clause $j$ as $l_{j_1}^j,l_{j_2}^j$, and $l_{j_3}^j$ where $j_1 < j_2 < j_3$. We construct instance $J$ of {\sc BeneficialManipulation} to have:
	
	\noindent{\bf Types:} $s+1$ types.
	
	\noindent{\bf Agents:}
	\begin{itemize}[wide,labelindent=0pt]
		\item For every variable $i\le s$, and every clause $j\le t$, two agents $0_i^j, 1_i^j$, and a dummy agent $d_i^j$.
		\item For every clause $j$, an agent $c_j$.
		\item A special agent $0$.
	\end{itemize}
	
	\noindent{\bf Items:} For every agent $a$ and every type $k\le s+1$, an item named $[a]_k$.
	
	\noindent{\bf Preferences:} For some agents, we only specify their importance orders (or local preferences) over types (or items) that are important for this proof, and assume that their preferences are an arbitrary linear order with the specified preferences over the top few types (or items).
	\begin{itemize}[wide,labelindent=0pt]
		\item {\bf agent $0$} has importance order $s+1 \rhd \others$, with local preferences:
		\begin{itemize}[labelindent=5pt]
			\item type $s+1$: $[0]_{s+1} \succ [c_1]_{s+1} \succ \others$
			\item every other type $k< s+1$: $[0]_{k} \succ \others$
		\end{itemize}
		
		\item {\bf agents $l_i^j$}, $l \in \{0,1\}$ have importance order $i \rhd \others$.
		\begin{itemize}[labelindent=5pt]
			\item type $i$: $\nxt_i^j \succ l_i^j \succ 0_i \succ \others$, where $\nxt_i^j = [l_i^{j+1}]_i$ if $j < t$, and $\nxt_i^j = [l_i^1]_i$ if $j=t$.
			\item type $s+1$ preferences are conditioned on assignment on type $i$.
			\begin{itemize}
				\item $D_i\setminus\{\nxt_i^j$\}: $[d_i^j]_i \succ \others$.
				\item $\nxt_i^j$: $[l_i^j]_i \succ \others$.
			\end{itemize}
			\item every other type $k$: $[l_i^j]_k \succ \others$.
		\end{itemize}
		
		\item {\bf agents $d_i^j$} have importance order $i \rhd \others$.
		\begin{itemize}[labelindent=5pt]
			\item type $i$: $[d_i^j]_i \succ \others$.
		\end{itemize}

		\item {\bf agents $c_j$} have importance order $s+1 \rhd \others$.
		\begin{itemize}[labelindent=5pt]
			\item type $s+1$: $[l_{j1}^j]_{s+1} \succ [l_{j2}^j]_{s+1} \succ [l_{j3}^j]_{s+1} \succ [0]_{s+1} \succ \others$.
			\item every other type $k$: $[c_j]_k \succ \others$.
		\end{itemize}
	\end{itemize}
	
	\noindent{\bf Sequential mechanism:} composed of serial dictatorships applied in the order $O=1 \rhd \dots \rhd s+1$, where the priority orders over agents are:
	\begin{itemize}[wide,labelindent=0pt]
		\item for types $i \le s$: $(\others, 0, 0_i^t, \dots, 0_i^1, 1_i^t, \dots, 1_i^1)$.
		\item type $s+1$: $(0_1^1, \dots 0_1^t, 0_2^1, \dots, 0_s^t, 1_1^1, \dots, 1_1^t, 1_2^1, \dots, 1_s^t, c_1, \dots, c_t, 0,\allowbreak \others)$.
	\end{itemize}
	Similar to preferences, we only specify the part of the priority orderings over the agents for each serial dictatorship that is relevant to the proof, and assume that the priority orderings are linear orders over the agents, where the specified part holds.
	
	The main idea is that if the 3-SAT instance is satisfiable, special agent $0$ enables every $c_j$ agent to get an item of type $s+1$ corresponding to a literal $l_i^j$ that satisfies the clause $j$ by a beneficial manipulation which results in agents $l_i^j$ corresponding to literals in clause $j$ being allocated their favorite item of type $i$. 
	
	When agents report preferences truthfully and are either optimistic or pessimistic, it is easy to check $f_O$ allocates items as follows: for types $i<s+1$ agent $0$ gets $0_i$ and every agent $l_i^j$ gets $\nxt_i^j$. Then, for type $s+1$, for any $l\in\{0,1\}$, every agent $l_i^j$ gets the item $[l_i^j]_{s+1}$. This in turn makes it so that for every $i\le s, j\le t$, the items $[l_i^j]_{s+1}$ unavailable to the agent $c_j$. Then, $c_1$ gets $[0]_{s+1}$, and finally, $0$ gets $[c_1]_{s+1}$. 
	
	Upon examining agent $0$'s preferences, it is easy to check that the only way for agent $0$ to improve upon this allocation is to receive a better item of type $s+1$, specifically, item $[0]_{s+1}$.
	
	\noindent$\Rightarrow$	Let $\phi$ be a satisfying assignment for instance $I$. Consider the manipulation where agent $0$ reports her top item of type $i$ to be $[0_i^1]_i$ if $\phi_i = 0$, and $[1_i^1]_i$ if $\phi_i = 1$.
	Now, suppose that every other agent reports preferences truthfully. 
	
	Let us consider the case where for some $i\le s$, $\phi_i=0$. It is easy to check that for type $i$, agents' allocations are as follows: 
	Agent $0$ gets $[0_i^1]_i$ if $\phi_i=0$, and in the sequence $j=t\dots2$ agents $0_i^j$ get items $[0_i^j]_i$ respectively, and agent $0_i^1$ gets $[0]_i$, i.e. none of the agents $0_i^j$ gets their corresponding top item $\nxt_i^j$. 
	Now, for type $s+1$, agent $0_i^j$ gets $[d_i^j]_{s+1}$ according to their true preferences since they did not receive their item $\nxt_i^j$ of type $i$, leaving item $[0_i^j]_{s+1}$ available. Then, agents $1_i^j$ get the items $[1_i^j]_{s+1}$, crucially, before agents $c_j$ get to choose any item. Then for every agent $c_j$, if $0_{i^*}^j$ is the literal with the lowest index $i^*$ such that $\phi_{i^*}$ corresponds to a satisfying assignment of clause $c_j$, $[0_{i^*}^j]_{s+1}$ must be available when $c_j$ gets her turn to pick an item, and gets it. Moreover, since $\phi$ is a satisfying assignment, there is such an item for every $c_j$. This leaves $[0]_{s+1}$ available when it is agent $0$'s turn to pick an item. Thus, special agent $0$ prefers the resulting allocation to her allocation when she picked items truthfully, and the manipulation was beneficial, irrespective of whether agent $0$ is optimistic or pessimistic.
	
	\noindent$\Leftarrow$ Suppose agent $0$ has a beneficial manipulation. Then, as we have already established, agent $0$ must get item $[0]_{s+1}$ as a result of the manipulation. Now, agents $c_j$ get their turn to pick an item before agent $0$ in the serial dictatorship for type $s+1$. Then, since they are truthful, each agent $c_j$ receives an item $[l_i^j]_{s+1}$ corresponding to a satisfying assignment. Otherwise one of them must get $[0]_{s+1}$, a contradiction. 
	
	Let us construct an assignment $\phi$ as follows: if $c_j$ gets item $[0_{i^*}^j]_{s+1}$ in the final allocation, set $\phi_i = 0$, and set $\phi_i =1$ otherwise. We will show that $\phi$ is a satisfying assignment for $I$. Since agents $0_i^j$ and $1_i^j$ come before agents $c_j$ in the serial dictatorship, and are also truthful, it must be that for every item $[l_{i^*}^j]_{s+1}$ allocated to agent $c_j$ in the final allocation, the corresponding agent $l_{i^*}^j$ does not get $\nxt_{i^*}^j$ of type $i^*$. 
	
	By construction of the preferences over type $i^*$ and the serial dictatorship for type $i^*$, it must be that special agent $0$ picks an item $[l_{i^*}^{k}]_{i^*}$, where either $k > j$ or $k = 1$. It is easy to check from the construction that if this is not the case, agent $l_i^j$ can pick item $\nxt_i^j$ when every agent other than $0$ picks truthfully. Further, if agent $0$ picks some item $[0_{i^*}^{\hat{j}}]_{i^*}$, it is easy to check that by the construction of the preferences, every agent other than $1_{i^*}^j$ gets their top item. Thus, none of the agents $c_j$ may receive the item $[1_{i^*}^j]$ since it must already have been picked by the agent $1_{i^*}^j$ in the serial dictatorship. Together with the fact that every agent $c_j$ receives an item that corresponds to a satisfying assignment, $\phi$ constructed above is a satisfying assignment for instance $I$.
\end{proof}

\end{document}